\documentclass[aps,superscriptaddress,twocolumn,showpacs,pra]{revtex4-2}

\usepackage{amsmath}
\usepackage{latexsym}
\usepackage{amssymb}
\usepackage{graphicx}
\usepackage[colorlinks=true, citecolor=blue, urlcolor=blue]{hyperref}
\usepackage{float}
\usepackage{amsfonts}
\usepackage{textcomp}
\usepackage{mathpazo}
\usepackage{comment}
\usepackage{xr}

\usepackage{bbm}

\usepackage{xcolor}
\definecolor{myurlcolor}{rgb}{1,.5,0}
\definecolor{mycitecolor}{rgb}{0,0,1}
\definecolor{myrefcolor}{rgb}{1,.5,0}
\usepackage{hyperref}
\hypersetup{colorlinks,
linkcolor=myrefcolor,
citecolor=mycitecolor,
urlcolor=myurlcolor}

\usepackage[draft]{fixme}
\usepackage{amsmath,bbm}
\usepackage{graphicx}
\usepackage{amsfonts}
\usepackage{amssymb}
\usepackage{amsmath, amssymb, amsthm,verbatim,graphicx,bbm}
\usepackage{mathrsfs}
\usepackage{color,xcolor,longtable}

\usepackage{xr}
\makeatletter
\newcommand*{\addFileDependency}[1]{
  \typeout{(#1)}
  \@addtofilelist{#1}
  \IfFileExists{#1}{}{\typeout{No file #1.}}
}
\makeatother

\newcommand{\beq}[0]{\begin{equation}}
\newcommand{\eeq}[0]{\end{equation}}

\newcommand{\one}{\leavevmode\hbox{\small1\normalsize\kern-.33em1}}

\def\be{\begin{equation}}
\def\ee{\end{equation}}
\def\ben{\begin{eqnarray}}
\def\een{\end{eqnarray}}
\def\eea{\end{array}}
\def\bea{\begin{array}}

\newcommand{\Tr}[1]{\mathrm{Tr}#1}
\newcommand{\bei}{\begin{itemize}}
\newcommand{\eei}{\end{itemize}}
\newcommand{\ket}[1]{|#1\rangle}
\newcommand{\bra}[1]{\langle#1|}

\newcommand{\proj}[1]{\ket{#1}\!\!\bra{#1}}

\newcommand{\I}{\mathbbm{1}}

\renewcommand{\emph}[1]{\t\mathrm{ext}bf{#1}}

\makeatletter
\newtheorem*{rep@theorem}{\rep@title}
\newcommand{\newreptheorem}[2]{%
\newenvironment{rep#1}[1]{%
 \def\rep@title{#2 \ref{##1}}%
 \begin{rep@theorem}}%
 {\end{rep@theorem}}}
\makeatother

\theoremstyle{plain}
\newtheorem{thm}{Theorem}
\newtheorem*{thm*}{Theorem}
\newreptheorem{thm}{Theorem}

\newtheorem{cor}[thm]{Corollary}

\theoremstyle{definition}

\theoremstyle{remark}

\usepackage[T1]{fontenc}

\begin{document}

\title{Detecting quantum resources in a semi-device independent framework }
\author{Shubhayan Sarkar}
\email{shubhayan.sarkar@ulb.be}
\affiliation{Laboratoire d’Information Quantique, Université libre de Bruxelles (ULB), Av. F. D. Roosevelt 50, 1050 Bruxelles, Belgium}
\author{Chandan Datta}
\email{cdatta@iitj.ac.in}
\affiliation{Institute for Theoretical Physics III, Heinrich Heine University D\"{u}sseldorf, Universit\"{a}tsstra{\ss}e 1, D-40225 Düsseldorf, Germany}
\affiliation{Department of Physics, Indian Institute of Technology Jodhpur, Jodhpur 342030, India}

\begin{abstract}	
We investigate whether one can detect the presence of a quantum resource in some operational task or equivalently whether every quantum resource provides an advantage over its free counterpart in some black box scenarios where one does not have much information about the devices. For any dimension $d$, we find that for any resource theory with less than $d^2$ number of linearly independent free states or free operations, there exist correlations that can detect the presence of a quantum resource. For this purpose, we introduce the framework for detecting quantum resources semi-device independently by considering the prepare-and-measure scenario with the restriction on the dimension of the quantum channel connecting the preparation box with the measurement box. We then explicitly construct witnesses to observe the presence of various quantum resources. We expect these results will open avenues for detecting and finding uses of quantum resources in general operational tasks.

\end{abstract}

\maketitle

{\em Introduction--}
 With the emergence of quantum technologies in the last few years, it has become clear that entanglement is an indispensable resource for a range of intriguing technological applications that cannot be matched by classical means \cite{Entanglement_review}. The applications of entanglement are practically endless \cite{Entanglement_review}, spanning from 
 quantum teleportation \cite{Teleportation}, to the secure distribution of cryptographic keys among several parties \cite{QKD}. However, it has become clear in recent years that entanglement is not the only resource for quantum technological applications, and there are other applications that rely on other features of quantum systems, such as coherence \cite{Baumgratz_coherence,Coherence_review}, imaginarity \cite{Hickey_2018,imaginarity_PRL,imaginarity_PRA,Varun_2023}, asymmetry \cite{asymmetry}, contextuality \cite{contextuality,Nature_contextuality,Review_contextuality}, purity \cite{purity_horodecki,purity_gour,purity_streltsov}, and so on.

To investigate the roles of these resources in a wide range of information-processing tasks, we must first understand their characteristics. Quantum resource theories \cite{HORODECKI_2012, Review_QRT} proposed in quantum information science provide a unified approach for studying all such resources, as well as their roles and limitations in various quantum technological applications. While the exploration of entanglement characterisation began in the early 1990s \cite{Teleportation,LOCC_1,LOCC_2,LOCC_3,Entanglement_review}, it was not until 2008 when Gour and Spekkens introduced the concept of resource theory and established a comprehensive framework in their investigation of the resource theory of asymmetry \cite{asymmetry}. Soon after, a number of additional resource theories—including coherence \cite{Baumgratz_coherence,Coherence_review}, imaginarity \cite{Hickey_2018,imaginarity_PRL,imaginarity_PRA,Varun_2023}, purity \cite{purity_horodecki,purity_gour,purity_streltsov}, stabiliser quantum computation \cite{Veitch_2014}, quantum measurements \cite{Oszmaniec_2019}, and numerous others—were introduced in the literature \cite{Review_QRT}. As a consequence, a generic framework for quantum resource theories has been established \cite{Coecke_2016,Fritz_2017,Review_QRT}.

After identifying all of these resources, it is critical to observe their actual presence in practical scenarios. One of the most common methods for observing this is state tomography. However, in general, this method is not economical as it involves reconstructing the complete information about the given quantum state. Moreover, one also has to trust their measurement device. Similarly, in the subchannel discrimination task, it was shown in \cite{CD1} that any resource state belonging to a convex resource theory provides an advantage over the free states. This result was then generalised to dynamical resource theories \cite{CD4}, generalised probabilistic theories \cite{CD3}, and non-convex resource theories \cite{CD2}. However, in all these works, one must trust the subchannels along with the quantum state. 

Although the above works show that quantum resources show an advantage over their free counterpart, the element of trust in these scenarios makes the problem of detecting resources not applicable in their formalism. For instance, if one already knows the quantum state then there is no point in detecting the resource inside the device.
In this work, we pose a stronger question. Can quantum states or operations be detected to be resourceful without trusting them or any other operation involved in the test? To put it simply, we enquire whether one can operationally detect the resource of unknown quantum states and operations.
For instance, violation of Bell inequalities serves as a witness for entanglement detection \cite{Nonlocality_review} where one does not trust any state or measurements. Similarly, assuming multiple independent sources, it was shown in \cite{Renou_2021,sarkar2025} that one can detect whether any of the quantum states or measurements need to be imaginary, thus providing a way to detect imaginarity without trusting any state or measurements.


Here, we focus on one of the well-known semi-device independent scenarios known as the prepare-and-measure (PM), where one assumes an upper bound on the dimension of the quantum channel connecting the preparation and measurement \cite{Gallego_2010}. In a practical scenario, it can be imagined as a communication task involving two separate boxes, a preparation box responsible for creating quantum states, which are then transmitted to a measurement box. 
A fascinating fact of this scenario is that we can perceive those boxes as black boxes, so removing the need for trust in them. The prepare-and-measure scenario has been employed in the literature for a wide range of applications, including random access codes \cite{RAC1,RAC2}, informational principles \cite{Pawlowski_2009,Chaves_2015}, key distribution \cite{QKD1,QKD2}, dimension witness \cite{Gallego_2010,Brunner_2013}, self-testing \cite{selftesting_review}, and more.

As we are concerned here with detecting resourceful states and operations, we treat operations as generalised instruments that can generate quantum or classical outputs. 
Restricting the quantum channel to transmit $d-$dimensional states, we then show that for any resource theory with strictly less than $d^2$ number of linearly independent free states or free operations, one can always detect whether the preparation box produces at least one resource state and the operation box performs at least one resourceful operation. Then considering particular resource theories, we construct explicit witnesses to detect their presence inside the preparation as well as the measurement box.


{\em Quantum resource theories--} 
Any quantum resource theory is built on two basic components: a set of free states and a set of free operations \cite{Review_QRT}. The quantum states that do not belong to the set of free states are referred to as resource states. Any resource theory must adhere to the fundamental constraint that the set of free operations must not be able to generate resource states from the set of free states. Hence, quantum resource theories allow us to study quantum information processing tasks under a restricted set of operations. As an example, in the resource theory of entanglement, separable states denote free states, local operations and classical communications (LOCC) denote free operations, and information processing tasks are studied under LOCC \cite{Entanglement_review}. 

Let us now restrict to states and operations acting on $\mathbb{C}^d$ for any $d\geq2$, where $d$ is the dimension of the Hilbert space.
The free states are denoted as $\sigma_d$ and the collection of these states is denoted by the set $\mathcal{S}_d$. Now, consider a set of linearly independent density matrices $\proj{\Gamma_{d,i}}\in\mathcal{S}_d$ using which any state $\sigma_d$ can be written via their linear combination. Equivalently, the elements of the set $ \mathcal{S}_{\mathrm{ext},d}=\{\proj{\Gamma_{d,i}}\}_i$ forms a basis of $\mathcal{S}_d$. However, it is important to note here that $\mathcal{S}_d$ might not contain all the density matrices which can be expressed as linear combinations of $\proj{\Gamma_{d,i}}$. Moreover, it is also important to realise that not every linear combination of these states would result in a valid density matrix. Furthermore, the set  $ \mathcal{S}_{\mathrm{ext},d}$ might not be unique. Let us now denote the cardinality of $\mathcal{S}_{\mathrm{ext},d}$ as $n(\mathcal{S}_{\mathrm{ext},d})$.
For instance, in the resource theory of coherence \cite{Coherence_review}, the set $\mathcal{S}_{\mathrm{ext},d}$ is given by $\mathcal{S}_{\mathrm{ext},d}=\{\proj{0},\proj{1},\ldots,\proj{d-1}\}$ and thus $n(\mathcal{S}_{\mathrm{ext},d})=d$. Similarly, for the resource theory of athermality or asymmetry \cite{athermality_brandao,asymmetry}, we have that $n(\mathcal{S}_{\mathrm{ext},d})\leq(d-1)^2+1$. Likewise, for the resource theory of magic or imaginarity in dimension two \cite{Veitch_2014,imaginarity_PRL}, we have that $n(\mathcal{S}_{\mathrm{ext},2})=3$. For any dimension $d$, the maximum cardinality of this set in quantum theory is $n_{\mathrm{max}}(\mathcal{S}_{\mathrm{ext},d})=d^2$ \cite{D_Ariano_2005}.

Similarly, one can also characterize free operations $\delta_d$ based on the above procedure. Equivalently, we express the operation $\delta_d$ using its Kraus decomposition as $\delta_d\equiv \{K_{j,\delta_d}^{\dagger}K_{j,\delta_d}\}$ which might not be unique. The operator $K_{j,\delta_d}^{\dagger}K_{j,\delta_d}$ corresponds to the positive operator-valued measure (POVM) element of the operation $\delta_d$. Let us denote the set of these free operations as $\mathcal{O}_d$ which is composed of all possible POVM elements of every free operation acting on $\mathbb{C}^d$. Let us now consider a set of linearly independent positive operators $\proj{\Delta_{d,i}}\in\mathcal{O}_d$ such that they form a basis of $\mathcal{O}_d$. Let us denote the set of these linearly independent operators as $\mathcal{O}_{\mathrm{ext},d}=\{\proj{\Delta_{d,i}}\}_i$. The cardinality of $\mathcal{O}_{\mathrm{ext},d}$ is denoted as $n(\mathcal{O}_{\mathrm{ext},d})$. Any operator in $\mathcal{O}_d$ can be expressed as a linear combination of elements in $\mathcal{O}_{\mathrm{ext},d}$. For instance, in coherence resource theory $n(\mathcal{O}_{\mathrm{ext},d})=d$ \cite{Coherence_review}, for the resource theory of asymmetry $n(\mathcal{O}_{\mathrm{ext},d})\leq (d-1)^2+1$ \cite{asymmetry}, and for imaginarity in dimension two \cite{imaginarity_PRL}, $n(\mathcal{O}_{\mathrm{ext},2})=3$. Again, in quantum theory for any dimension $d: n_{\mathrm{max}}(\mathcal{O}_{\mathrm{ext},d})=d^2$. Any resource theory defined on a $d-$dimensional Hilbert space, from now on will be denoted as $\{\mathcal{S}_d,\mathcal{O}_d\}$.

\begin{figure}
    \centering
    \includegraphics[width=\linewidth]{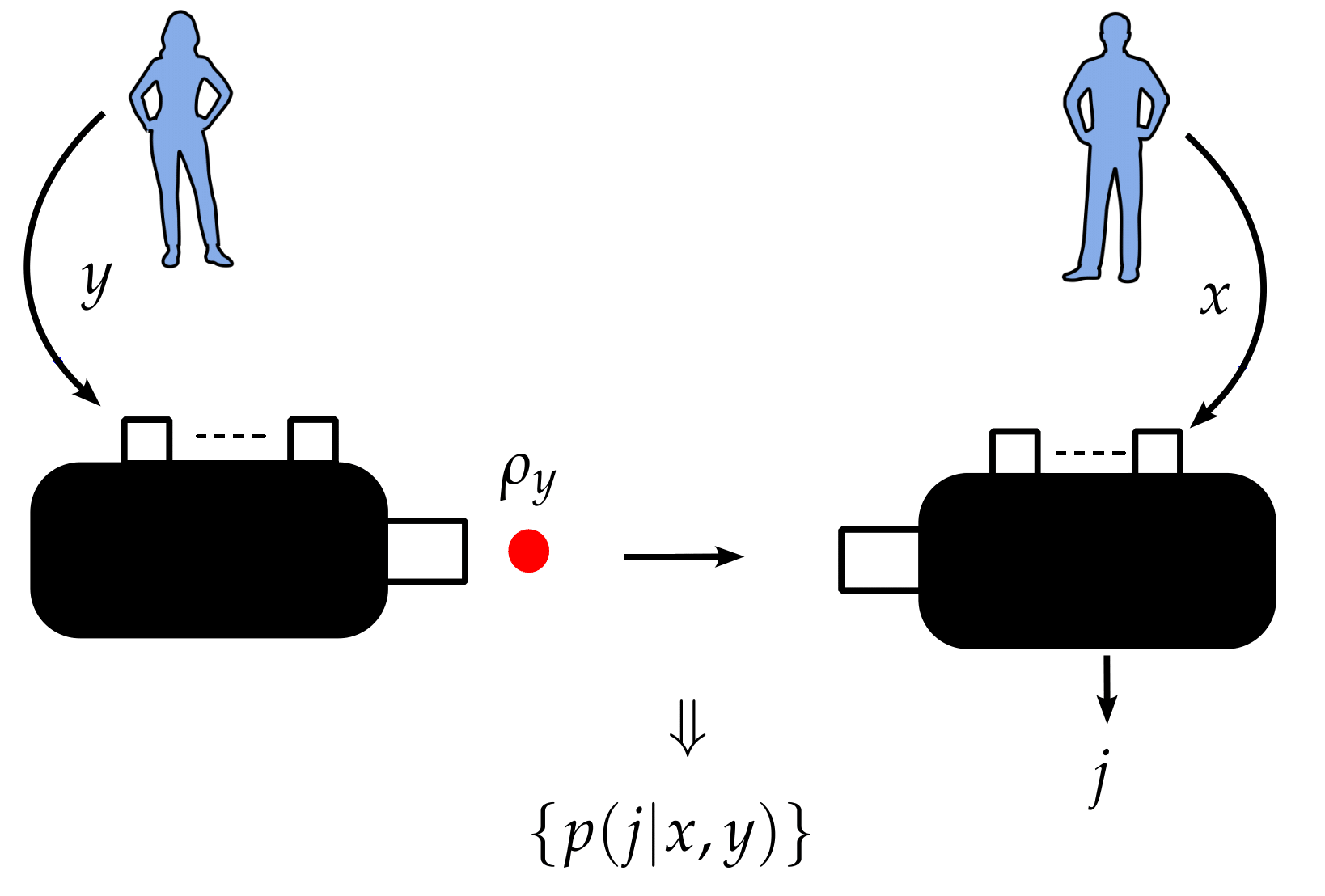}
    \caption{ We consider two parties, a producer and an observer, each situated in spatially separated labs. They can communicate solely via a quantum channel supporting a quantum system of dimension $d$. The producer freely selects $y$ inputs to prepare quantum states $\rho_y$, which are then sent to the observer. The observer, in turn, selects $x$ inputs to apply quantum operations on the received states, yielding outcomes $j$. After all the runs of the experiment, one obtains the probability distribution $\{p(j|x,y)\}$.}
    \label{fig1}
\end{figure}

{\em Prepare-and-measure scenario--} 
We consider here a prepare-and-measure (PM) scenario with two parties, producer and observer in two different spatially separated labs who can communicate only via a quantum channel that can support a quantum system bounded by the dimension $d$. The producer can freely choose $y$ inputs based on which it prepares quantum states $\rho_y$ and sends it to the observer. The observer now freely chooses $x$ inputs based on which quantum operations (generalised quantum instruments) are applied on the incoming quantum state to obtain the outcomes $j$. This is illustrated in Fig. \ref{fig1}. 

From the above setup, one obtains the probability distribution $\Vec{p}=\{p(j|x,y)\}$ where $p(j|x,y)$ denotes the probability of obtaining outcome $j$ by the operation box given the inputs $x,y$. In quantum theory, the inputs $y$ generate a qudit state $\rho_y$ acts on $\mathbb{C}^d$ and the operation box does the operation denoted by the POVM elements as $\mathbb{O}_x=\{K_{j,x}^{\dagger}K_{j,x}\}$ such that they act on $\mathbb{C}^d$. Now, the probabilities $p(j|x,y)$ are given by
\begin{equation}
    p(j|x,y)=\Tr(K_{j,x}^{\dagger}K_{j,x}\rho_y).
\end{equation}
Notice that in the above description, we treat operations in the same way as measurements in quantum theory. In general, it may not be possible to obtain these probabilities in a practical setting. However, we show below that treating operations similar to measurement allows us to detect the presence of a quantum resource in a semi-device-independent way. To detect resources in an experiment, one can restrict to measurements and all below analysis follows straightaway.
Furthermore, the correlations $\vec{p}$ in general are not convex due to the dimension constraint. Consequently, finding general results for semi-device independent schemes is highly non-trivial. 

{\em Results--} Using the above framework, let us now establish two no-go theorems that apply to a large class of quantum resource theories. The first one concerns the case when the set of linearly independent free states is strictly less than $d^2$.

\begin{thm}\label{theo1} Consider a resource theory $\{\mathcal{S}_d,\mathcal{O}_d\}$ such that the cardinality of the $\mathcal{S}_{\mathrm{ext},d}$ is $N$ with $N<d^2$. Consider now the PM scenario described above with $y=0,1,\ldots,N$ and $x=0,1,\ldots,d^2-1$ and any $j\geq2$ such that one obtains $\vec{p}=\{p(j|x,y)\}$. Then, there exists at least one probability distribution $\vec{p}_{Q}$ in quantum theory that can not be realized by free states in $\mathcal{S}_d$. 
\end{thm}

The proof of the above theorem can be found in Appendix A of \cite{SupMat}. Let us extend the above theorem for resource theories with the number of linearly independent free operations strictly less than $d^2$.

\begin{thm}\label{theo2} Consider a resource theory $\{\mathcal{S}_d,\mathcal{O}_d\}$ such that the cardinality of the $\mathcal{O}_{\mathrm{ext},d}$ is $N$ such that $N<d^2$. Consider now the PM scenario described above with $y=0,1,\ldots,d^2-1$ and $x=0,1,\ldots,N$ and any $j\geq2$ such that one obtains $\vec{p}=\{p(j|x,y)\}$. Then, there exists at least one probability distribution $\vec{p}_{Q}$ in quantum theory that can not be realized by free operations in $\mathcal{O}_d$.
\end{thm}

We refer to Appendix A of the \cite{SupMat} for the proof of the aforementioned theorem. Theorem \ref{theo1} shows that one can always detect the presence of a resource state corresponding to the resource theory $\{\mathcal{S}_d,\mathcal{O}_d\}$ in the PM scenario if $n(\mathcal{S}_{\mathrm{ext},d})<d^2$. Similarly, Theorem \ref{theo2} shows that one can always detect the presence of a resource operation corresponding to the resource theory $\{\mathcal{S}_d,\mathcal{O}_d\}$ in the PM scenario if $n(\mathcal{O}_{\mathrm{ext},d})<d^2$. One can straightaway infer from the above results that if one wants to detect the presence of a resource state as well as operation simultaneously, then one needs to consider the PM scenario with $y,x=0,1,\ldots,d^2-1$. Thus, we obtain the following corollary of Theorems \ref{theo1} and \ref{theo2}.
\setcounter{thm}{0}
\begin{cor}\label{cor1} Consider a resource theory $\{\mathcal{S}_d,\mathcal{O}_d\}$ such that the cardinality of the $\mathcal{S}_{\mathrm{ext},d}$ and $\mathcal{O}_{\mathrm{ext},d}$ is less than $d^2$. Consider now the PM scenario described above with $x,y=0,1,\ldots,d^2-1$ and any $j\geq2$ such that one obtains $\vec{p}=\{p(j|x,y)\}$. Then, there exists at least one probability distribution $\vec{p}_{Q}$ in quantum theory that can not be realized by either the free states in $\mathcal{S}_d$ or operations in $\mathcal{O}_d$. Thus, observing $\vec{p}_Q$ allows one to conclude that at least one of the preparations and operations is resourceful.
\end{cor}

Let us remark here that the number of inputs $x,y$ in the above corollary might not be optimal. As we will see below for particular resource theories, there might also exist probability distributions with a lower number of inputs $x,y$ such that at least one of the preparations and operations is resourceful. For this purpose, we construct witnesses which is a function of the probability distribution $\vec{p}$ of the form
\begin{eqnarray}
    W_{\mathcal{S}_d,\mathcal{O}_d}=f(\vec{p})\leq \beta_{\mathcal{S}_d,\mathcal{O}_d}
\end{eqnarray}
where $\beta_{\mathcal{S}_d,\mathcal{O}_d}$ is the maximum value attainable using the free states and operations in the resource theory $\{\mathcal{S}_d,\mathcal{O}_d\}$. Consequently, the presence of a resource can be verified if the above bound is violated. Furthermore, from a practical perspective it is usually beneficial to construct witnesses as one can never exactly obtain the theorised correlations in an experiment and witnesses provide a gap that can tolerate such errors.  

{\em Witnesses--} Let us begin by giving a generic construction of witnesses that follows from Corollary \ref{cor1}. For this purpose, we consider the probability distribution $\vec{p}_Q$ which can not be attained by free states and operations. As $x,y=0,1,\ldots,d^2-1$, the witness is given by
\begin{eqnarray}
    W_{\mathcal{S}_d,\mathcal{O}_d}=-\sum_{x,y=0}^{d^2-1}|p(0|x,y)-p_Q(0|x,y)|\leq \beta_{\mathcal{S}_d,\mathcal{O}_d}
\end{eqnarray}
where $\beta_{\mathcal{S}_d,\mathcal{O}_d}=-\varepsilon$ and $\varepsilon\rightarrow0^+$. The value of the above functional, when at least one state and operation is resourceful, is $0$. Consequently, the gap between the values obtainable using free states and operations and quantum theory is close to $0$ and thus highly susceptible to experimental errors. 

Let us now construct witnesses tailored to specific resource theories that can detect the presence of concerned quantum resources that provide a considerable gap between the maximum value obtainable using free states and operations to quantum theory. 
In the first example, we construct a witness for the resource theory of coherence \cite{Baumgratz_coherence, Coherence_review} in the PM scenario when the channel is limited to transmit qubits. In this theory free states and free operations are known as incoherent states and incoherent operations (IO), respectively. Here, we consider the PM scenario with $x=0,1$ and $y=0,1,2$. In this scenario, the witness to observe coherence can be constructed as 
\begin{eqnarray}\label{coherwit}
     W_C= p(0|0,0)+p(0|0,1)&+&p(0|1,0)+p(1|1,1)\nonumber\\&+&p(1|0,2)\leq \beta_{C},
\end{eqnarray}
where $\beta_{C}$ can be evaluated by considering incoherent states and operations \cite{Coherence_review}. The Kraus operator of any incoherent operation $\{K_j\}$ is given by $K_j=\sum_{i=0}^{d-1}c_{i,j}\ket{f(i)}\!\bra{i}$ where $f(i)$ is a function with domain and range being the labels of the incoherent basis vectors and thus the corresponding POVM is given by $K_j^{\dagger}K_j=\sum_{i=0}^{d-1}|c_{i,j}|^2\ket{i}\!\bra{i}$. 
Consequently, restricting to $d=2$ we find that $\beta_{C}=4$. Any violation of this bound would imply that at least one of the states and operations is coherent. Now, consider the following states,
\begin{eqnarray}\label{incostate}
    \rho_0=\proj{0},\qquad\rho_1=\proj{+},\qquad \rho_2=\proj{\overline{1}},
\end{eqnarray}
and measurements,
\begin{eqnarray}\label{incomea}
    K_{0,0}&=&\proj{\overline{0}},\qquad K_{1,0}=\proj{\overline{1}},\nonumber\\K_{0,1}&=&\proj{\overline{+}},\qquad K_{1,0}=\proj{\overline{-}},
\end{eqnarray}
where $\{\ket{\overline{0}},\ket{\overline{1}}\}$ are the eigenvectors of $\frac{\sigma_x+\sigma_z}{\sqrt{2}}$ and $\{\ket{\overline{+}},\ket{\overline{-}}\}$ are the eigenvectors of $\frac{\sigma_x-\sigma_z}{\sqrt{2}}$. Here, $\sigma_x$ and $\sigma_z$ denote respective Pauli matrices. Then, we obtain that $W_C=1+4\cos^2(\pi/8)=3+\sqrt{2}$, which violates the upper bound, i.e., $\beta_{C}=4$. Since the resource theories of asymmetry \cite{asymmetry} and athermality \cite{athermality_brandao, Lostaglio_athermality} are equivalent to the coherence theory \cite{Baumgratz_coherence, Coherence_review} for $d=2$, the same finding can be readily applied to both theories. Moreover, we identify that the above witness $W_C$ in Eq. \eqref{coherwit} can be generalised to the case where the channel can transmit qudits. For this purpose, we consider the PM scenario in Fig. \ref{fig1} such that the number of inputs in the preparation box is $d^2$ and the operation box is two and is given as $y\equiv y_0y_1$ where $y_0,y_1=0,1,\ldots,d-1$ and $x=0,1$. Based on the quantum random access codes (QRAC) for arbitrary dimensions introduced in \cite{RAC2}, we represent the witness of coherence in arbitrary dimensions as
\begin{eqnarray}
    W_{C,d}=\sum_{x=0,1}\sum_{y_0,y_1=0}^{d-1}p(j=y_x|x,y_0y_1)\leq\beta_{C,d}.
\end{eqnarray}
where $\beta_{C,d}$ is the maximum value of $W_{C,d}$ attainable using incoherent states and operations. In the supplementary materials, we show that $\beta_{C,d}=d^2+d$. Consider now the following observables, $Z_d=\sum_{i=0}^{d-1}\omega^{i}\proj{i},\ X_d=\sum_{i=0}^{d-1}\ket{i}\!\bra{i+1}$ which are $d-$dimensional generalisation of the Pauli $Z,X$ observables respectively with $\omega=e^{2\pi i/d}$. Then, putting the states $\ket{\psi_{y_0y_1}}=X_d^{y_0}Z_d^{y_1}\ket{\psi_{00}}$ where $\ket{\psi_{00}}=\frac{1}{\sqrt{2\sqrt{d}(1+\sqrt{d})}}\Big((\sqrt{d}+1)\ket{0}+\sum_{i=1}^{d-1}\ket{i}\Big)$ along with the operations $Z_d,X_d$ corresponding to input $x=0,1$, gives us $W_{C,d}=d^2+d\sqrt{d}$. This further shows us that the resource that gives the quantum advantage in QRAC is quantum coherence.

In the next example, we consider resource theory of imaginarity \cite{imaginarity_PRL, imaginarity_PRA} and construct a witness for the detection of imaginarity. Again, we take the PM scenario with $x=0,1,2$ and $y=0,1,2,3$ and $d=2$, and the witness to observe imaginarity can be designed as
\begin{eqnarray}\label{witness_imaginarity}
    W_I=W_C+p(0|2,3)-|p(0|2,0)-p(1|2,0)|\nonumber\\-|p(0|2,1)-p(1|2,1)|\leq \beta_{I}.
\end{eqnarray}
where $\beta_{I}$ is the maximal value attainable considering real states and operations \cite{imaginarity_PRL}. Consequently, to evaluate $\beta_I$, we recall that any state $\rho\in \mathbb{C}^2$ can be expressed as $\rho=(\I+r_x\sigma_x+r_y\sigma_y+r_z\sigma_z)/2$, where $\sigma_i$ are Pauli matrices and $r_i$'s are components of the Bloch vector $\vec{r}$. Since in the resource theory of imaginarity free states satisfy $\rho=\rho^*$, one can readily obtain that any free state can be expressed as $\rho=1/2(\I+r_z\sigma_z+r_x\sigma_x)$. Similarly, for the POVM elements, we have $K_j^{\dagger}K_j=m_0\I+m_z\sigma_z+m_x\sigma_x$. Using these representations we evaluate that $\beta_{I}\leq 5$. Any violation of this bound would imply that at least one of the states and measurements must be complex. Now, consider the following states 
\begin{eqnarray}
    \rho_0&=&\proj{0},\qquad\rho_1=\proj{+},\nonumber\\ \rho_2&=&\proj{\overline{0}},\qquad\rho_3=\proj{+_y},
\end{eqnarray}
and measurements
\begin{eqnarray}
    K_{0,0}&=&\proj{\overline{0}},\quad K_{1,0}=\proj{\overline{1}},\quad  K_{0,1}=\proj{\overline{+}},\nonumber\\ K_{1,0}&=&\proj{\overline{-}},\quad  K_{0,2}=\proj{+_y},\quad K_{1,2}=\proj{-_y},\ \ \ \ 
\end{eqnarray}
where $\ket{\pm_y}=1/\sqrt{2}(\ket{0}\pm i\ket{1})$.
Using the above states and operations one can obtain $W_I=4+\sqrt{2}$, which violates the upper bound given in Eq. (\ref{witness_imaginarity}). 
In Appendix B of \cite{SupMat}, we also find specific witnesses for resource theory of purity and magic.

{\em Discussions--}
In the above work, we considered the prepare-and-measure scenario where the measurement box can perform generalised quantum instruments. By limiting the quantum channel to transmit $d$-dimensional states, we showed that for any resource theory, be it convex or non-convex, with fewer than $d^2$ linearly independent free states or operations, it is always possible to detect if the preparation box produces at least one resource state and if the operation box performs at least one resourceful operation. 
While our focus in this letter has been on single-partite systems, it is worth noting that our method can be readily extended to multipartite systems, thus encompassing the resource theory of entanglement.

In subchannel discrimination task \cite{CD1}, it was shown that for every convex resource theory, the advantage obtained is directly related to the robustness measure of the state, which followed from the fact that all the quantum resources share the same linear witness. In our case, such a general relation can not be obtained due to the simple fact that one might not always obtain a linear witness to observe the advantage in the PM scenario [For instance, see the witness of imaginarity in Eq. \eqref{witness_imaginarity}]. Moreover, to detect resources of unknown states or operations in the PM scenario, one would require more than one state, operations and witnesses to detect resources which in general will be different for different resource theories. Furthermore, the optimisation again has to be carried over the states (unlike subchannel discrimation), so the advantage can not be related to the robustness measure of some states. However, it remains an interesting open problem whether the advantage obtained is related to some operational quantity. This would further address the problem of finding sufficient conditions to violate a witness in the PM scenario.

A practical application of our results can be illustrated in the following scenario: Imagine an experimental lab intending to purchase a preparation box guaranteed to produce at least one resource state, and an operation box that performs at least one resourceful operation. Suppose the experimenter has access to a quantum channel capable of supporting at max qudit systems. In that case, they can verify the capabilities of these boxes without needing to trust them or have any prior knowledge about their workings.

Our work opens up various interesting questions in the PM scenario. As previously mentioned in Fig. \ref{fig1}, our objective was to confirm the presence of resources within the two boxes. Now, consider a situation in which the boxes are intended to generate resource states and operations only. It would be interesting to verify if the PM scenario can confirm whether the boxes are functioning as intended, i.e., they are not producing any free state or operation. The following is another intriguing follow-up question to consider: Consider now a theory in dimension $d$ where free states are given by the convex combination of $d^2$ linearly independent vectors. An instance of such linearly independent vectors is provided in \cite{D_Ariano_2005}. Let us for the time being assume that it is a valid resource theory. 
Even if such a set of free states would not include every quantum state in dimension $d$, it is not clear whether such a resource would be detected using the above results in the PM scenario.

\begin{acknowledgements}
We thank Alexander Streltsov for his useful insights. S.S. acknowledges the QuantERA II Programme (VERIqTAS project) that has received funding from the European Union’s Horizon 2020 research and innovation programme under Grant Agreement No 101017733. C.D. acknowledges support from the German Federal Ministry of Education and Research (BMBF) within the funding program ``quantum technologies -- from basic research to market'' in the joint project QSolid (grant number 13N16163).
    
\end{acknowledgements}

%

\onecolumngrid
\appendix
\section{Proofs}
\setcounter{thm}{0}
\begin{thm}\label{theo1}Consider a resource theory $\{\mathcal{S}_d,\mathcal{O}_d\}$ such that the cardinality of the $\mathcal{S}_{\mathrm{ext},d}$ is $N$ with $N<d^2$. Consider now the PM scenario described above with $y=0,1,\ldots,N$ and $x=0,1,\ldots,d^2-1$ and any $j\geq2$ such that one obtains $\vec{p}=\{p(j|x,y)\}$. Then, there exists at least one probability distribution $\vec{p}_{Q}$ in quantum theory that can not be realized by free states in $\mathcal{S}_d$. 
\end{thm}
\begin{proof} The proof is by contradiction. Let us first identify $\vec{p}_Q$ by considering that for each $y$ the preparation box produces the corresponding pure states $\ket{\psi_y}\in \mathbb{C}^d$ such that $\{\proj{\psi_y}\}_y$ are linearly independent. Similarly, the operations are chosen such that for each input $x$ the corresponding Kraus operators are $K_{0,x}=\proj{\phi_x}$ such that $\{\proj{\phi_x}\}_x$ are again linearly independent. As $x=0,1,\ldots,d^2-1$, the set $\{\proj{\phi_x}\}$ forms a basis for any operator acting on $\mathbb{C}^d$. Now, any probability $p_Q(0|x,y)$ realised by these states and operators are given by $ p_Q(0|x,y)=|\bra{\psi_y}\phi_x\rangle|^2$. Moreover, as the set $\{\proj{\phi_x}\}$ is a basis, $ p_Q(0|x,y)\ne0$ for at least one $x$ for every $y$.  From here on, we refer to $\vec{p}_Q=\{p_Q(0|x,y)\}$ for the rest of the proof.
    
Let us now assume that the vector $\vec{p}_{Q}$ can be realized using free states in $\mathcal{S}_d$. As described earlier, any free state in $\mathcal{S}_d$ can be written as a linear combination of $\proj{\Gamma_{d,i}}\in\mathcal{S}_{\mathrm{ext},d}$ and thus we obtain that
    \begin{equation}\label{2}
   \sum_{i=1}^N \lambda_{i,y} \bra{\Gamma_{d,i}}\tilde{K}_{0,x}^{\dagger}\tilde{K}_{0,x}\ket{\Gamma_{d,i}}=p_Q(0|x,y)\quad \forall x,y
    \end{equation}
     where we do not put any restrictions on the Kraus operators $\tilde{K}_{0,x}$ and $\lambda_{i,y}$ are some scalars such that $\sum_{i=1}^{N}\lambda_{i,y}\proj{\Gamma_{d,i}}$ are valid quantum states for all $y$.  

    We now show that the above relations given in Eq. \eqref{2} can not hold if $n(\mathcal{S}_{\mathrm{ext},d})=N<d^2$. For this purpose, we observe that the above relations in Eq. \eqref{2} can be expressed as a matrix equation of the form $M_F=M_Q$, where  $M_F$ and $M_Q$ are matrices with $N+1$ rows and $d^2$ columns such that
    \begin{eqnarray}
        M_F&=&\sum_{y=0}^{N}\sum_{x=0}^{d^2-1}\left(\sum_{i=1}^N \lambda_{i,y} \bra{\Gamma_{d,i}}\tilde{K}_{0,x}^{\dagger}\tilde{K}_{0,x}\ket{\Gamma_{d,i}}\right)\ket{y}\!\bra{x},\nonumber\\
        M_Q&=&\sum_{y=0}^{N}\sum_{x=0}^{d^2-1}p_Q(0|x,y)\ket{y}\!\bra{x}.
    \end{eqnarray}
    Let us evaluate the rank of these matrices. We first consider $M_Q$ and show that its rank is $N+1$ and thus all the rows must be linearly independent. To see this, we take a linear combination of all rows to obtain the following condition
    \begin{eqnarray}\label{4}
        \sum_{y}\alpha_y\Tr(\proj{\psi_y}\ \proj{\phi_{x}})=0\quad\forall x
    \end{eqnarray}
    where $\alpha_y$ are some scalars and then we show that $\alpha_y=0$ for all $y$.
    Consequently, we take a linear combination of the above relations \eqref{4} that should satisfy for any $\{\beta_x\}$ the following condition
    \begin{eqnarray}
           \Tr\left(\sum_y\alpha_y\proj{\psi_y}\sum_x \beta_x\proj{\phi_{x}}\right)=0.
    \end{eqnarray}
    As $\{\proj{\phi_{x}}\}$ spans the space of all density matrices of dimension $d^2$, we obtain that
    \begin{eqnarray}
        \Tr\left(\sum_y\alpha_y\proj{\psi_y}\ \rho\right)=0\qquad \forall \rho\in\mathbb{C}^d.
    \end{eqnarray}
    Thus, 
    the only solution of the above formula is
    \begin{eqnarray}
        \sum_y\alpha_y\proj{\psi_y}=0
    \end{eqnarray}
   which using the fact that $\proj{\psi_y}$ are linearly independent gives us $\alpha_y=0$ for all $y$. Thus, all the rows of $M_Q$ are linearly independent, and $\mathrm{rank}(M_Q)=N+1$.

    Now, let us focus on $M_F$ and show that all its rows are not linearly independent. As done above, we take a linear combination of all rows of $M_F$ to obtain the following condition
    \begin{eqnarray}
        \sum_{i,y}\alpha_y\lambda_{i,y}\Tr(\proj{\Gamma_{d,i}}\tilde{K}_{0,x}^{\dagger}\tilde{K}_{0,x})=0\quad \forall x
    \end{eqnarray}
    which can be expressed as
    \begin{eqnarray}
    \sum_{i=1}^N\left(\sum_{y=0}^{N}\alpha_y\lambda_{i,y}\right)\Tr(\proj{\Gamma_{d,i}}\tilde{K}_{0,x}^{\dagger}\tilde{K}_{0,x})=0\quad \forall x.
    \end{eqnarray}
    If all the rows of $M_F$ are linearly independent, then the only solution of the above formula should be $\alpha_y=0$ for all $y$. However, the following condition satisfies the above formula 
    \begin{eqnarray}
        \sum_{y=0}^{N}\alpha_y\lambda_{i,y}=0\quad \forall i.
    \end{eqnarray}
    Given any fixed set of $\lambda_{i,y}$, there are $N$ equations with $N+1$ variables $\alpha_y$. Thus, one has infinite solutions of $\alpha_y$. Thus, the rows of $M_F$ are not linearly independent which allows us to conclude that $\mathrm{rank}(M_F)\leq N$. Consequently, our initial assumption $M_F=M_Q$ is incorrect as both matrices are of different ranks. This completes the proof.
\end{proof}

\begin{thm}Consider a resource theory $\{\mathcal{S}_d,\mathcal{O}_d\}$ such that the cardinality of the $\mathcal{O}_{\mathrm{ext},d}$ is $N$ such that $N<d^2$. Consider now the PM scenario described above with $y=0,1,\ldots,d^2-1$ and $x=0,1,\ldots,N$ and any $j\geq2$ such that one obtains $\vec{p}=\{p(j|x,y)\}$. Then, there exists at least one probability distribution $\vec{p}_{Q}$ in quantum theory that can not be realized by free operations in $\mathcal{O}_d$.
\end{thm}

\begin{proof} 
Again, the proof is by contradiction. Let us first identify $\vec{p}_Q$ by considering that for each $y$ the preparation box produces the corresponding pure states $\ket{\psi_y}\in \mathbb{C}^d$ such that $\{\proj{\psi_y}\}_y$ are linearly independent. As $y=0,1,\ldots,d^2-1$, the set $\{\proj{\psi_y}\}$ forms a basis for any operator acting on $\mathbb{C}^d$. Similarly, the operations are chosen such that for each input $x$ the corresponding Kraus operators are $K_{0,x}=\proj{\phi_x}$ such that $\{\proj{\phi_x}\}_x$ are again linearly independent. Now, any probability $p_Q(0|x,y)$ realised by these states and operators are given by $ p_Q(0|x,y)=|\bra{\psi_y}\phi_x\rangle|^2$. Moreover, as the set $\{\proj{\psi_y}\}$ is a basis, $ p_Q(0|x,y)\ne0$ for at least one $y$ for every $x$. From here on, we refer to $\vec{p}_Q=\{p_Q(0|x,y)\}$ for the rest of the proof.
    
Let us now assume that the vector $\vec{p}_{Q}$ can be realized using free operations in $\mathcal{O}$. As described earlier, any free operation in $\mathcal{O}_d$ can be written as a linear combination of $\proj{\Delta_{d,i}}\in\mathcal{O}_{\mathrm{ext},d}$ and thus we obtain that
    \begin{equation}\label{21}
   \sum_{i=1}^N \lambda_{i,x} \bra{\Delta_{d,i}}\tilde{\rho}_y\ket{\Delta_{d,i}}=p_Q(0|x,y)\quad \forall y
    \end{equation}
    where we do not put any restrictions on the state $\tilde{\rho}_{y}$ and $\lambda_{i,x}$ are some scalars such that $\sum_{i=1}^{N}\lambda_{i,x}\proj{\Delta_{d,i}}$ are valid quantum operators for all $y$..

    Let us now show that the above relations in Eq. \eqref{21} cannot hold if $n(\mathcal{O}_{\mathrm{ext}})=N<d^2$. We again observe that the above relations in Eq. \eqref{21} can be expressed as a matrix equation of the form $M_F'=M_Q'$, where  $M_F',M_Q'$ are matrices with $N+1$ rows and $d^2$ columns such that
    \begin{eqnarray}
       M_F'&=&\sum_{y=0}^{d^2-1}\sum_{x=0}^{N}\left(\sum_{i=1}^N \lambda_{i,y} \bra{\Delta_{d,i}}\tilde{\rho}_y\ket{\Delta_{d,i}}\right)\ket{x}\!\bra{y},\nonumber\\
        M_Q'&=&\sum_{y=0}^{d^2-1}\sum_{x=0}^{N}p_Q(0|x,y)\ket{x}\!\bra{y}.
    \end{eqnarray}
    
    As done in Theorem \ref{theo1}, one can straightaway show that $M_F'\ne M_Q'$ by showing that $\mathrm{rank}(M_Q')=N+1$ but $\mathrm{rank}(M_F')\leq N$ and thus proving that the relations in Eq. \eqref{21} can not hold. This completes the proof.
\end{proof}

\section{Witnesses for some more resource theories}
\subsection{Coherence for arbitrary dimension}
The witness to detect coherent states and operations is given by
\begin{eqnarray}
    W_{C,d}=\sum_{x=0,1}\sum_{y_0,y_1=0}^{d-1}p(j=y_x|x,y_0y_1)\leq \beta_{C,d}.
\end{eqnarray}
Let us now evaluate $\beta_{C,d}$. For this purpose, we express the above witness as
\begin{eqnarray}\label{13sup}
    W_{C,d}=\sum_{y_0,y_1=0}^{d-1}\Tr(\rho_{y_0y_1}(A_{y_0}+B_{y_1}))
\end{eqnarray}
where $\{A_{y_0}\},\{B_{y_1}\}$ are the measurement operators corresponding to input $x=0,1$ respectively and consequently we have that $\sum_{y_0}A_{y_0}=\sum_{y_1}B_{y_1}=\I_d$. Note that for $d-$ (or less) outcome measurements, this condition imposes that they must be rank-one. Now the expression \eqref{13sup} can be upper bounded as
\begin{eqnarray}\label{14sup}
     W_{C,d}\leq\sum_{y_0,y_1=0}^{d-1}||A_{y_0}+B_{y_1}||
\end{eqnarray}
where $||.||$ denotes the operator norm. Let us now notice that the function on the right-hand side of the above equation is convex over the measurements $\{A_{y_0}\},\{B_{y_1}\}$. Moreover, the incoherent measurements are rank-one and thus belong to a convex set with the extremal measurements: $\{\proj{i}\}$ and all its relabelling. From Bauer's maximum principle, it is known that for any convex function defined on a set that is convex and compact, attains its maximum at some extreme point of that set. Consequently, we can now maximize \eqref{14sup} by considering only the extremal measurements. A simple assignment of the extremal measurements to \eqref{14sup} gives us that $\beta_{C,d}=d^2+d$ which can be realised by choosing $A_{i}=B_{i}=\proj{i}$.

\subsection{Purity}
Here, we construct a witness for the resource theory of purity for any $d$ \cite{purity_horodecki,purity_gour,purity_streltsov}. In this theory, the free state is $\I/d$ and free operations are called unital operations $\Lambda_U$ which preserve the maximally mixed state, i.e., $\Lambda_U (\I/d)=\I/d$. We consider the PM scenario with $x,y=0$, and the witness to observe purity is simply given by 
\begin{eqnarray}
    W_P=p(0|0,0)\leq \beta_{P}.
\end{eqnarray}
It is straightforward to observe that the above witness is bounded by $1/d$ for $\rho_0=\I/d$ and any quantum operation as
\begin{eqnarray}
    \beta_{F,P}=p(0|0,0)=\Tr(K_{0,0}^\dagger K_{0,0} \rho_0)\leq\frac{1}{d}.
\end{eqnarray}
It is clear that one can obtain $W_p=1$ with $\ket{\psi_0}=\ket{0}$ and $ K_{0,0}=\proj{0}$. However, using the above witness one can not distinguish unital to non-unital operations. We leave that problem open for future explorations.

\subsection{Magic}

In the resource theory of magic, free states and free operations are stabiliser states and Clifford unitaries, respectively \cite{Veitch_2014,Howard_2017}. Interestingly for $d=2$, the witness $W_C$ in Eq. (4) of the main text also serves as a witness of magic resource given by
\begin{eqnarray}
     W_C= p(0|0,0)+p(0|0,1)&+&p(0|1,0)+p(1|1,1)+p(1|0,2)\leq \beta_{M}.
\end{eqnarray}
However, the maximal value that can be achieved by non-magic states $\beta_{M}$ is higher than $\beta_{C}$. Numerically evaluating $\beta_{M}$ gives us $\beta_{M}=4.32$. Here, the optimisation is carried over free states and any quantum measurement. Consequently, any value of $W_C$ beyond $\beta_M$ detects that the preparation box prepares at least one magic state. Consider now the states and Kraus operators given in Eqs. (5) and (6) of the main text, respectively, and as described in the main text one can obtain $W_C=3+\sqrt{2}$ for those states and measurements.


\end{document}